\def\be{\begin{equation}}
	\def\ee{\end{equation}}
\def\ba{\begin{array}}
	\def\ea{\end{array}}

\documentclass[prl,showpacs,twocolumn,amsmath]{revtex4}
\usepackage{amsmath}
\usepackage{amsfonts}
\usepackage{mathrsfs}
\usepackage{amssymb}
\usepackage{pifont}
\usepackage{epsfig,subfigure,dsfont,amsthm,amsbsy,mathrsfs,amscd}
\usepackage{epstopdf}
\usepackage{bbm}
\usepackage{color}     
\def\qed{\leavevmode\unskip\penalty9999 \hbox{}\nobreak\hfil
	\quad\hbox{\leavevmode  \hbox to.77778em{%
			\hfil\vrule   \vbox to.675em%
			{\hrule width.6em\vfil\hrule}\vrule\hfil}}
	\par\vskip3pt}
\usepackage{leftidx}
\newtheorem{theorem}{Theorem}

\input amssym.def
\begin{document}
	\title{\large\bf  Correlations Between Quantum Battery Capacity and Quantum Resources for Two-qubit System}
	
	\author{Yiding Wang, Xiaofen Huang and Tinggui Zhang$^{\dag}$}
	\affiliation{ School of Mathematics and Statistics, Hainan Normal University, Haikou, 571158, China \\
		$^{\dag}$ tinggui333@163.com}
	
	\bigskip
	\bigskip
	
	\begin{abstract}
	We investigate the relationship between quantum battery capacity and quantum resources in a two-qubit system consisting of mutually coupled battery and charger subsystems. We find that the battery capacity decreases monotonically with the quantum entanglement, steering, Bell nonlocality and coherence, and peaks when these four quantum resources vanish. Moreover, we reveal the capacity gap between the total system capacity and the sum of the battery and charger spin capacities, which is the residual battery capacity, and establish its positive correlation with entanglement. Furthermore, unlike the first four resources, although the battery capacity decreases monotonically with quantum imaginarity, its disappearance under system detuning does not guarantee a peak capacity, and this effect becomes more pronounced as the detuning increases. In contrast to the first five resources, the quantum state texture shows a positive correlation with battery capacity, but a negative correlation with entanglement, steering, Bell nonlocality, coherence, imaginarity, and residual battery capacity. These monotonic relationships are independent of the choice of system parameters. Our findings reveal the relationship between quantum battery capacity and quantum resources during the dynamic evolution of a quantum battery system, and advances the theory of quantum batteries and the development of quantum energy storage systems.
	\end{abstract}
	
	 \maketitle
	
	\section{I. Introduction}
	Quantum thermodynamics is a rapidly evolving field that focuses on understanding the behavior of quantum systems at the nanoscale \cite{jgmh}. A significant offshoot is the concept of quantum batteries, which leverage quantum effects to enhance performance characteristics, such as during the charging process \cite{ramf,gmam,dfgm,drgm,drgm2}. The inherent quantum effects in these batteries have significant potential to surpass the capabilities of classical batteries.
	
	The concept of quantum battery was first introduced by Alicki and Fannes \cite{ramf}, who investigated the role of quantum resources, including entanglement and coherence, in enabling an efficient work extraction from quantum systems. Following this pioneering work, quantum batteries have received significant attention, with considerable effort devoted to exploring this emerging concept. These efforts can be broadly categorized into two types. The first focuses on the interplay between quantum battery performance and various quantum resources, revealing that figures of merit such as ergotropy and charging power are closely linked to the quantum properties of the battery states \cite{kvhm,mpkv,amar,lpga,gffc,fhkf,gf,mbaa,hlss,hyyh,gmfp,wljc,strss,agcs,imad,rcrn}. The second aspect focuses on investigating the impact of specific models on the charging process of quantum battery \cite{gmam,drgm,jxlh,tplj,sgtc,fqdh,rgds,lfmp,yyzt,fqdy,dfgm,gmdf,yvda,dlyf,fmyf,lwsq2,sczz,zglg}. This includes implementations such as spin-chain models \cite{drgm,tplj,sgtc,fqdh,rgds}, Tavis-Cummings models \cite{gmam,jxlh}, harmonic oscillator models \cite{dfgm,gmdf,yvda}, Dicke models \cite{lfmp,yyzt,fqdy}, three-level systems \cite{dlyf,fmyf}, and cavity-optomechanical model \cite{lwsq2}. Collectively, the primary goal of these efforts is to demonstrate how systems and operations in quantum mechanics enhance battery performance. 
	
	Furthermore, because any quantum system inevitably interacts with its environment, dissipative charging of quantum batteries has been extensively studied \cite{fb,fttf,fhkf2,mcac,jqqw,wctr,mlsl,fclm,lwsq,wlsh,mbah,sqll}. For a comprehensive overview of research on quantum batteries, see \cite{fcsg}. Recent research has yielded new findings and advances in the field of quantum batteries \cite{gmav,acvc,hyyk,mlht,yyxs,jtzc,wlsj,fmyfq,cadm}. Andolina et al. in \cite{gmav} proposed a quantum battery model exhibiting a quantum advantage, which saturates the quantum speed limit. The model consists of two harmonic oscillators coupled via nonlinear interaction during nonequilibrium charging. Hu et al. \cite{mlht} investigated a wireless charging scheme for quantum batteries using $n$ charger units. They found that multiple charger units enhance charging performance in the weak and moderate coupling regimes, whereas their efficiency is lower in the strong coupling regime. Song et al. \cite{wlsj} proposed a quantum battery scheme based on the nitrogen-vacancy center in diamond, and revealed that the inherent robustness against self-discharge can be enhanced by increasing the proportion of coherent ergotropy.
	
	Quantum battery capacity, a new figure of merit for evaluating the performance of quantum batteries, was first introduced by Yang et al. in 2023 \cite{yyas}, and this physical quantity has recently been experimentally verified based on an optical platform \cite{xyyh}. Despite the significant interest it sparked \cite{tgzh,aasa,yhst,ywxh,hltz,ykwl}, several important questions remain to be addressed. As previously mentioned, quantum resources such as entanglement and coherence play a crucial role in quantum batteries. Understanding how these resources relate to battery capacity thus presents an interesting and valuable question. Elucidating this connection will not only advance our understanding of how quantum resources influence the quantum energy storage systems, but also provide theoretical guidance for designing nanoscale energy devices with genuine quantum advantages. In this study, we establish a set of relations between the battery state capacity and various quantum resources (including entanglement, steering, Bell nonlocality, coherence, negativity, and quantum state texture) within a quantum system comprising a battery spin and a charger spin, during its dynamical evolution. We find that with the exception of quantum state texture, the other five resources exhibit a trade-off relationship with the battery state capacity. These relationships are established under a specific model of mutual coupling between the battery and charger subsystems, but they are independent of the choice of system parameters.
	
	The remainder of this paper is organized as follows. In Section II, we introduce several physical quantities involved in this study and their corresponding quantitative methods, including quantum entanglement, quantum steering, Bell nonlocality, quantum coherence, quantum imaginarity, quantum state texture, and quantum battery capacity. In Section III, we present our main results: the relationship between quantum resources and quantum battery capacity, positive correlation between entanglement and residual battery capacity, and negative correlation between quantum state texture and residual battery capacity. We summarize and discuss our conclusions in the last section.

	\section{II. PRELIMINARIES}
   First, we provide a brief review of the quantum resources explored in this work, including quantum entanglement, quantum steering, Bell nonlocality, quantum coherence, quantum imaginarity, and quantum state texture. The mixedness of a reduced density matrix derived from a pure state indicates the presence of entanglement of composite system. Hence the mixedness in the subsystem $\rho_b$ directly reflects the entanglement between the two qubits. We can quantify this entanglement using concurrence, defined as: \cite{wkw,csf,rphk,oggt,mcc}:
   \begin{equation}\label{e1}
   	\mathcal{E}(\rho)=\sqrt{2[1-\text{Tr}(\rho_b^2)]}.
   \end{equation}

   Quantum steering is considered a quantum resource that is stronger than entanglement, but weaker than Bell nonlocality \cite{mtqt}. From the local hidden-state model, steering inequalities have been developed, and their violation demonstrates the presence of steering. Specifically, Cavalcanti et al. proposed a linear steering inequality to verify whether a bipartite state is steerable when both Alice and Bob perform three binary measurements on their respective subsystems. The form of three-setting linear steering inequality can be written as \cite{egcs}:
   \begin{equation}\label{e18}
   	F_3(\rho_{AB},\mu)=\frac{1}{\sqrt{3}}|\sum_{k=1}^{3}\langle A_kB_k\rangle|\leq1,
   \end{equation}
   where $A_k=\vec{a}_k.\vec{\sigma}$ and $B_k=\vec{b}_k.\vec{\sigma}$ with $\vec{a}_k, \vec{b}_k\in\mathbb{R}^3$ are the unit and orthonormal vectors, respectively, and $\vec{\sigma}=(\sigma_1, \sigma_2, \sigma_3)$. $\langle A_kB_k\rangle=\text{Tr}(\rho_{AB}A_kB_k)$, and $\mu=\{\vec{a}_1, \vec{a}_2, \vec{a}_3, \vec{b}_1, \vec{b}_2, \vec{b}_3\}$ is the set of measurement directions. The bipartite state $\rho$ is $F_3$ steerable if \cite{acsc,bpkm}
   \begin{equation}\label{e19}
   	\mathcal{S}(\rho)=\text{Tr}(T(\rho)T(\rho)^\tau)-1>0,
   \end{equation}
   where $T(\rho)=[t_{ij}]$ is the correlation matrix with $t_{ij}=\text{Tr}(\rho(\sigma_i\otimes\sigma_j))$, and $\tau$ represents the transposition operation. We use $\mathcal{S}(\rho)$ as the steering measure.
   
   For an arbitrary two-quibt state $\rho$, the maximum violation of the CHSH inequality is given by \cite{hhh}
   \begin{equation}\label{e23}
   	\mathcal{B}_{max}(\rho)=2\sqrt{\iota_1+\iota_2}.
   \end{equation}
   Here, $\iota_1$ and $\iota_2$ are the two largest  eigenvalues of $T(\rho)T(\rho)^\tau$, where $T(\rho)$ is the correlation matrix. In classical scenarios, the CHSH value does not exceed 2, which allows us to use the violation of the CHSH inequality to quantify the quantum resource of Bell nonlocality, i.e.
   \begin{equation}\label{e24}
   	\mathcal{B}(\rho)=\mathcal{B}_{max}(\rho)-2=2\sqrt{\iota_1+\iota_2}-2.
   \end{equation}
   
   Quantum coherence characterizes the ability of quantum systems to maintain phase relationships, serving as a central resource in quantum computing and information processing. Among all coherence monotones, the measure induced by the $l_1$ norm serves as a vital tool for quantifying coherence, which can be denoted as \cite{tbmc}
   \begin{equation}\label{c1}
   	C_1(\rho)=\min\limits_{\delta\in\mathcal{I}}\|\rho-\delta\|_{l_1}=\sum_{i\neq j}|\rho_{ij}|,
   \end{equation}
   where the minimum state $\delta_\text{min}=\Delta(\rho)$ is directly connected to the original state $\hat{\rho}$ by physical operation dephasing $\Delta$. In other words, the density matrix of $\delta_\text{min}$ consists only of the diagonal elements of $\rho$
   
   Quantum imaginarity was first introduced by Hickey and Gour \cite{ahgg}. Although they established a rigorous resource theory, its utility remained unclear until Renou et al. found an experimental scenario that could not be modeled with only real-valued quantum amplitudes \cite{mord}. This finding has stimulated extensive research on quantum imaginarity. It was shown that imaginarity has significant effects on certain discrimination tasks \cite{kdwt}, machine learning \cite{msvs}, pseudorandomness \cite{thkb}, and weak-value theory \cite{rwef}. Here, we employ the $l_1$-norm to quantify this quantum resource. The $l_1$-norm of imaginarity for a quantum state $\rho$ is defined as the sum of the absolute values of the imaginarity parts of all off-diagonal elements in the density matrix \cite{qctg},
   \begin{equation}\label{e4}
   	\mathcal{I}(\rho)=\sum_{i\neq j}|\text{Im}(\rho_{ij})|.
   \end{equation}
   The $l_1$-norm of imaginarity for $\rho$ is zero under the condition that all off-diagonal elements of the density matrix are real in the chosen basis. Such states, with vanishing quantum imaginarity, are known as free states within the resource theory of imaginarity.
   
   Quantum state texture (QST) is a novel quantum resource that has recently been proposed \cite{fp,wyd,czbh,atpa,yzsk}, receiving considerable interest owing to its intimate relationship with quantum coherence. Under the given computational basis $\{|i\rangle\}$, we consider the row and column indices of the density matrix of a quantum state as the first two dimensions and the real or imaginary part of each matrix element as the third dimension. In this way, each quantum state is associated with a unique three-dimensional plot. The quantum state texture is what characterizes the uneveness of this three-dimensional plot.  There exists one and only one texture-free state, expressed as:
   \begin{equation}\label{e2}
   	f_1=|f_1\rangle\langle f_1|,
   \end{equation}
   where $|f_1\rangle=\frac{1}{\sqrt{d}}\sum_{i=0}^{d-1}|i\rangle$ and $d$ is the dimension of Hilbert space $H$. In this work, we choose the trace distance measure $\mathcal{T}_{tr}$ to quantify QST, defined as \cite{wyd}:
   \begin{equation}\label{e3}
   	\mathcal{T}_\text{tr}(\rho)=\frac{1}{2}\text{Tr}|\rho-f_1|,
   \end{equation}
   where $\text{Tr}|A|=\text{Tr}(AA^\dagger)^\frac{1}{2}$ is the trace norm (or 1-norm) of $A$.
   
   In this study, we focus on the relationship between quantum resources and quantum battery. The parameter that serve as key figure of merit for evaluating the performance of quantum batteries is the battery capacity \cite{yyas,xyyh,tgzh}, which is defined as the energy difference between the active and passive states corresponding to the battery state $\rho$:
   \begin{equation}\label{e5}
   	\mathcal{C}(\rho;H)=\text{Tr}[\rho^\uparrow H]-\text{Tr}[\rho^\downarrow H],
   \end{equation}
   where $\rho^\uparrow=U^\uparrow\rho(U^\uparrow)^\dagger$ is the active state associated with $\rho$, $\rho^\downarrow=U^\downarrow\rho(U^\downarrow)^\dagger$ is the passive state associated with $\rho$, $U^\uparrow$ and $U^\downarrow$ denote the unitary operations that realize the minimum and the maximum of the work extraction functional. A state is active if and only if no further energy can be injected into it by means of unitary operation. By this reasoning, if a state $\rho_b^\downarrow$ is passive, then no more energy can be extracted from it using unitary transformations. Note that an analytical expression for the battery capacity exists, based on the unitary invariance, in terms of the spectrum of the battery state and the energy levels of the Hamiltonian \cite{yyas}:
   \begin{equation}\label{e6}
   	\mathcal{C}(\rho;H)=\sum_{i=0}^{d-1}\epsilon_i(\lambda_i-\lambda_{d-1-i}),
   \end{equation}
   where $\{\lambda_i\}$ and $\{\epsilon_i\}$ represent the energy levels of $\rho$ and the Hamiltonian $H_b$, respectively, arranged in ascending order without loss of generality, i.e., $\lambda_0\leqslant\lambda_1\leqslant...\leqslant\lambda_{d-1}$ and $\epsilon_0\leqslant\epsilon_1\leqslant...\leqslant\epsilon_{d-1}$.

	\section{III. Quantum battery capacity and Quantum Resource}
	\subsection{ Quantum battery capacity}
	For a two-qubit quantum battery system composed of one battery spin and one charger spin, the Hamiltonian can be written as
	\begin{align}
		H=H_b&+H_c+H_I,\label{e7}\\
		H_b&=-\omega_b\sigma_z^b,\label{e8}\\
		H_c&=-\omega_c\sigma_z^c,\label{e9}\\
		H_I=J_1(\sigma_+^b\sigma_-^c&+\sigma_-^b\sigma_+^c)+J_2\sigma_z^b\sigma_z^c.\label{e10}
	\end{align}
	Here, $\sigma_z^b$ and $\sigma_z^c$	represent the spin operators of the battery and the charger parts, respectively, $\omega_b$ and $\omega_c$ denote the effective external magnetic field strengths for the battery and charger parts, and $\sigma_{x,y,z}^{b(c)}$ refer to the standard Pauli operators. $\sigma_\pm^{b(c)}=(\sigma_x^{b(c)}\pm i\sigma_y^{b(c)})/2$ represent the up and down operators of the qubit. $J_1$ and $J_2$ are the spin flip-flop interaction strength and Ising interaction strength, respectively.
		
    In the initial time, the battery is in the ground state $\rho_b(0)=|0\rangle_b$, and the charger is in the excited state $\rho_c(0)=|1\rangle_c$. The evolving state is given by $\rho(t)=U(t)\rho(0)U(t)^\dagger$ with $U(t)=\exp(-i\int_{0}^{t}Hdt)$. Herein, we set Planck constant $\hbar=1$. Since the interaction term $J_1(\sigma_+^b\sigma_-^c+\sigma_-^b\sigma_+^c)$ preserves the total excitation number, the evolution involves only the two states $|01\rangle$ and $|10\rangle$. We can therefore restrict the dynamics to a two-dimensional subspace and simplify the Hamiltonian accordingly. Specifically,
    \begin{equation*}
    	\begin{split}
    		\langle01|H|01\rangle&=\langle01|(-\omega_b\sigma_z^b)|01\rangle+\langle01|(-\omega_c\sigma_z^c)|01\rangle\\
    		&+\langle01|[J_1(\sigma_+^b\sigma_-^c+\sigma_-^b\sigma_+^c)+J_2\sigma_z^b\sigma_z^c]|01\rangle\\
    		&=\omega_b-\omega_c-J_2.
    	\end{split}
    \end{equation*}
    Here we have used that $\sigma_+^b\sigma_-^c|01\rangle=|10\rangle$ and $\sigma_-^b\sigma_+^c|01\rangle=0$. 
    \begin{equation*}
    	\begin{split}
    		\langle01|H|10\rangle&=\langle01|(-\omega_b\sigma_z^b)|10\rangle+\langle01|(-\omega_c\sigma_z^c)|10\rangle\\
    		&+\langle01|[J_1(\sigma_+^b\sigma_-^c+\sigma_-^b\sigma_+^c)+J_2\sigma_z^b\sigma_z^c]|10\rangle\\
    		&=J_1.
    	\end{split}
    \end{equation*}
    Repeating this process, we can obtain $\langle10|H|01\rangle=J_1, \langle10|H|10\rangle=\omega_c-\omega_b-J_2$. Therefore, the system Hamiltonian given in Eqs. (\ref{e7})-(\ref{e10}) can be rewritten as
    \begin{equation}\label{e11}
    	H=\left(
    	\begin{array}{cc}
    		\omega_b-\omega_c-J_2 & J_1\\
    		J_1 & \omega_c-\omega_b-J_2\\
    	\end{array}
    	\right ).
    \end{equation} 
    The eigenvalues of the Hamiltonian given by (\ref{e11}) are $e_1=\sqrt{J_1^2+(\omega_b-\omega_c)^2}-J_2, e_2=-\sqrt{J_1^2+(\omega_b-\omega_c)^2}-J_2$, with the corresponding eigenvectors are
    \begin{equation*}
    	\begin{split}
    		|e_1\rangle&=\frac{1}{\sqrt{\xi_1^2+1}}(\xi_1|01\rangle+|10\rangle),\\
    		|e_2\rangle&=\frac{1}{\sqrt{\xi_2^2+1}}(\xi_2|01\rangle+|10\rangle),
    	\end{split}
    \end{equation*}
    where $\xi_1=(\omega_b-\omega_c+\sqrt{J_1^2+(\omega_b-\omega_c)^2})/J_1$, and $\xi_2=(\omega_b-\omega_c-\sqrt{J_1^2+(\omega_b-\omega_c)^2})/J_1$. We express the initial state $|01\rangle$ in terms of the eigenvectors $|e_1\rangle$ and $|e_2\rangle$ as:
    \begin{equation*}
    	|01\rangle=\frac{\sqrt{\xi_1^2+1}}{\xi_1-\xi_2}|e_1\rangle-\frac{\sqrt{\xi_2^2+1}}{\xi_1-\xi_2}|e_2\rangle.
    \end{equation*}
    Thus, the state of the total system at time $t$ can be written as:
    \begin{equation*}
    	\begin{split}
    		|\psi(t)\rangle&=\exp(-iHt)|01\rangle\\
    		&=\frac{e^{-ie_1t}\sqrt{\xi_1^2+1}}{\xi_1-\xi_2}|e_1\rangle-\frac{e^{-ie_2t}\sqrt{\xi_2^2+1}}{\xi_1-\xi_2}|e_2\rangle\\
    		&=\frac{e^{-ie_1t}}{\xi_1-\xi_2}(\xi_1|01\rangle+|10\rangle)-\frac{e^{-ie_2t}}{\xi_1-\xi_2}(\xi_2|01\rangle+|10\rangle)\\
    		&=\frac{(e^{-ie_1t}\xi_1-e^{-ie_2t}\xi_2)|01\rangle+(e^{-ie_1t}-e^{-ie_2t})|10\rangle}{\xi_1-\xi_2}.
    	\end{split}
    \end{equation*}
    By performing the partial trace operation, the density matrix of the battery state $\rho_b(t)=\text{Tr}_c[|\psi(t)\rangle\langle\psi(t)|]$ is given by
    \begin{equation}\label{e12}
    	\rho_b(t)=\left(
    	\begin{array}{cc}
    		\frac{1+\cos(e_1-e_2)t+\frac{2(\omega_b-\omega_c)^2}{J_1^2}}{2+\frac{2(\omega_b-\omega_c)^2}{J_1^2}} & 0\\
    		0 & \frac{1-\cos(e_1-e_2)t}{2+\frac{2(\omega_b-\omega_c)^2}{J_1^2}}\\
    	\end{array}
    	\right ).
    \end{equation}
	
	Our idea is to reveal more fundamental relationships among entanglement $\mathcal{E}$, steering $\mathcal{S}$, Bell nonlocality $\mathcal{B}$, coherence $C_1$, imaginarity $\mathcal{I}$, QST $\mathcal{T}_{tr}$, and battery capacity $\mathcal{C}$, rather than focusing solely on physical interpretations under specific parameter choices.
	    
	\subsection{Quantum entanglement}
	     We begin by discussing entanglement, which develops between the battery and the charger during the time evolution of the total system. Here, we derive a relation between battery-charger entanglement and the capacity of battery state during dynamic evolution.
	    \begin{theorem}
	    During the dynamic evolution of the quantum battery system, the battery-charger entanglement $\mathcal{E}(\rho(t))$ and the battery state capacity $\mathcal{C}(\rho_b(t);H_b)$ at time $t$ satisfy the following relationship:
	    \begin{equation}\label{e13}
	    		\mathcal{C}(\rho_b(t);H_b)=2\omega_b\sqrt{1-\mathcal{E}(\rho(t))^2}.
	    \end{equation}
	    In other words, battery-charger entanglement is detrimental to the battery state capacity.
	    \end{theorem}
	    \begin{proof}
	    From the Eq. (\ref{e12}), the density matrix of the battery state at time $t$ can be parameterized as
	    \begin{equation*}
	    	\rho_b(t)=\left(
	    	\begin{array}{cc}
	    		p(t) & 0\\
	    		0 & 1-p(t)\\
	    	\end{array}
	    	\right ),
	    \end{equation*} 
	    where 
	    $$p(t)=\frac{1+\cos(e_1-e_2)t+\frac{2(\omega_b-\omega_c)^2}{J_1^2}}{2+\frac{2(\omega_b-\omega_c)^2}{J_1^2}}.$$
	    According to Eqs. (\ref{e1}) and (\ref{e6}), we can calculate that the capacity of the battery state $\mathcal{C}(\rho_b(t);H_b)=2\omega_b|1-2p(t)|$, and the entanglement
	    \begin{equation*}
	    \mathcal{E}(\rho(t))=2\sqrt{p(t)-p(t)^2}.
	    \end{equation*}
	    This means that 
	    \begin{equation*}
	    p(t)^2-p(t)+\frac{\mathcal{E}(\rho(t))^2}{4}=0.
	    \end{equation*}
	    The solutions to this quadratic equation are $p(t)=1/2\pm\sqrt{1-\mathcal{E}(\rho(t))^2}/2$. Substituting the solutions for $p(t)$ into the expression of $\mathcal{C}(\rho_b(t);H_b)$, we obtain that
	    \begin{equation*}
	    \begin{split}
	    \mathcal{C}(\rho_b(t);H_b)&=2\omega_b|1-2p(t)|\\
	    &=2\omega_b|1-2(\frac{1}{2}\pm\frac{\sqrt{1-\mathcal{E}(\rho(t))^2}}{2})|\\
	    &=2\omega_b|\pm\sqrt{1-\mathcal{E}(\rho(t))^2}|\\
	    &=2\omega_b\sqrt{1-\mathcal{E}(\rho(t))^2}.
	    \end{split}
	    \end{equation*}
	    \end{proof}
	    Eq. (\ref{e13}) indicates that during time evolution, the capacity of the battery state diminishes as the battery-charger entanglement increases. In fact, Theorem 1 establishes a mathematical connection between the battery-charger entanglement $\mathcal{E}$ and the capacity of the battery state $\mathcal{C}$. We aim to further understand the physical essence of this phenomenon. To this end, we first demonstrate that the sum of the battery spin capacity and the charger spin capacity is bounded by a fixed upper bound. 
	    \begin{theorem}
	    During the dynamic evolution of the quantum battery system, such a trade-off relationship holds:
	    \begin{equation}\label{e14}
	    	\mathcal{C}(\rho_b(t);H_b)+\mathcal{C}(\rho_c(t);H_c)\leq\mathcal{C}(\rho(t);H).
	    \end{equation}
	    \end{theorem}
	    See Appendix A for a detailed proof of Eq. (\ref{e14}). Following the idea in \cite{ywxh}, we can naturally define the residual battery capacity within the model considered in this work as the difference between the total system capacity and the subsystem capacity:
	    \begin{equation}\label{e15}
	    	R(\rho,H)=\mathcal{C}(\rho;H)-\mathcal{C}(\rho_b;H_b)-\mathcal{C}(\rho_c;H_c).
	    \end{equation}
	    If we consider $\mathcal{C}(\rho_b;H_b)$ and $\mathcal{C}(\rho_c;H_c)$ as a combined entity representing the subsystem battery capacity, then Eq. (\ref{e15}) reveals a trade-off relationship between the subsystem battery capacity and the residual battery capacity, in which the magnitude of one directly impacts the other. The following results demonstrate that the residual battery capacity and the battery-charger entanglement are positively correlated during the dynamical evolution.
	    \begin{theorem}
	    During the time evolution of the quantum battery system, the battery-charger entanglement $\mathcal{E}(\rho(t))$ and the residual battery capacity $R(\rho(t),H)$ at time $t$ satisfy:
	    	\begin{equation}\label{e16}
	    	\mathcal{R}(\rho(t),H)=\epsilon_4-\epsilon_1-2(\omega_b+\omega_c)\sqrt{1-\mathcal{E}(\rho(t))^2}.
	    	\end{equation}
	    Here, $\epsilon_4$ and $\epsilon_1$ denote the maximum and minimum eigenvalues, respectively, of the total system Hamiltonian.
	    \end{theorem}
	    \begin{proof}
	    We consider the spectral decomposition of the total system Hamiltonian, $H=\sum_{i=1}^4\epsilon_i|\epsilon_i\rangle\langle\epsilon_i|$, with its eigenvalues arranged in ascending order: $\epsilon_1\leq\epsilon_2\leq\epsilon_3\leq\epsilon_4$. Based on the proof of Theorem 1, we can obtain the battery capacity $\mathcal{C}(\rho_b(t);H_b)=2\omega_b|1-2p(t)|$. By symmetry, the capacity for the charger spin follows readily as $\mathcal{C}(\rho_c(t);H_c)=2\omega_c|1-2p(t)|$. Since the total battery capacity is invariant under unitary evolution, we have $\mathcal{C}(\rho(t);H)=\mathcal{C}(\rho(0);H)=(1-0)\epsilon_4+(0-1)\epsilon_1=\epsilon_4-\epsilon_1$. Therefore, the residual battery capacity at time $t$ is given by
	    \begin{equation*}
	    	\mathcal{R}(\rho(t),H)=\epsilon_4-\epsilon_1-2(\omega_b+\omega_c)|1-2p(t)|.
	    \end{equation*}
	    Note that the capacities of the battery and charger spins are related by:
	    \begin{equation*}
	    	\omega_c\mathcal{C}(\rho_b(t);H_b)=\omega_b\mathcal{C}(\rho_c(t);H_c).
	    \end{equation*}
	    Thus the expression for the capacity of the battery state can be written as
	    \begin{equation}\label{e17}
	    	\mathcal{C}(\rho_b(t);H_b)=\frac{\omega_b}{\omega_b+\omega_c}(\epsilon_4-\epsilon_1-\mathcal{R}(\rho(t),H)).
	    \end{equation}
	    Substituting Eq. (\ref{e17}) into Eq. (\ref{e13}) yields the relation between the residual battery capacity and entanglement. Since $\mathcal{E}(\rho(t))$ is negatively correlated with $\mathcal{C}(\rho_b(t);H_b)$, it follows that the residual battery capacity is positively correlated with the entanglement.
	    \end{proof}
	     When two subsystems are highly entangled, most of the battery capacity is residual. This occurs because the reduction operation discards a significant amount of quantum information, leading to a substantial gap between the battery capacity of the subsystems and that of the total system. Moreover, when entanglement is maximized, the subsystem battery capacity is compressed to zero, at which point all capacity becomes residual. This may represent the physical essence of how battery-charger entanglement adversely affects the capacity of battery states. The relationship between battery–charger entanglement, residual battery capacity $R(\rho,H)$, and battery state capacity $\mathcal{C}(\rho_b(t);H_b)$, is shown in Figure 1.
	    \begin{figure}[htbp]
	   	\centering
	   	\includegraphics[width=0.5\textwidth]{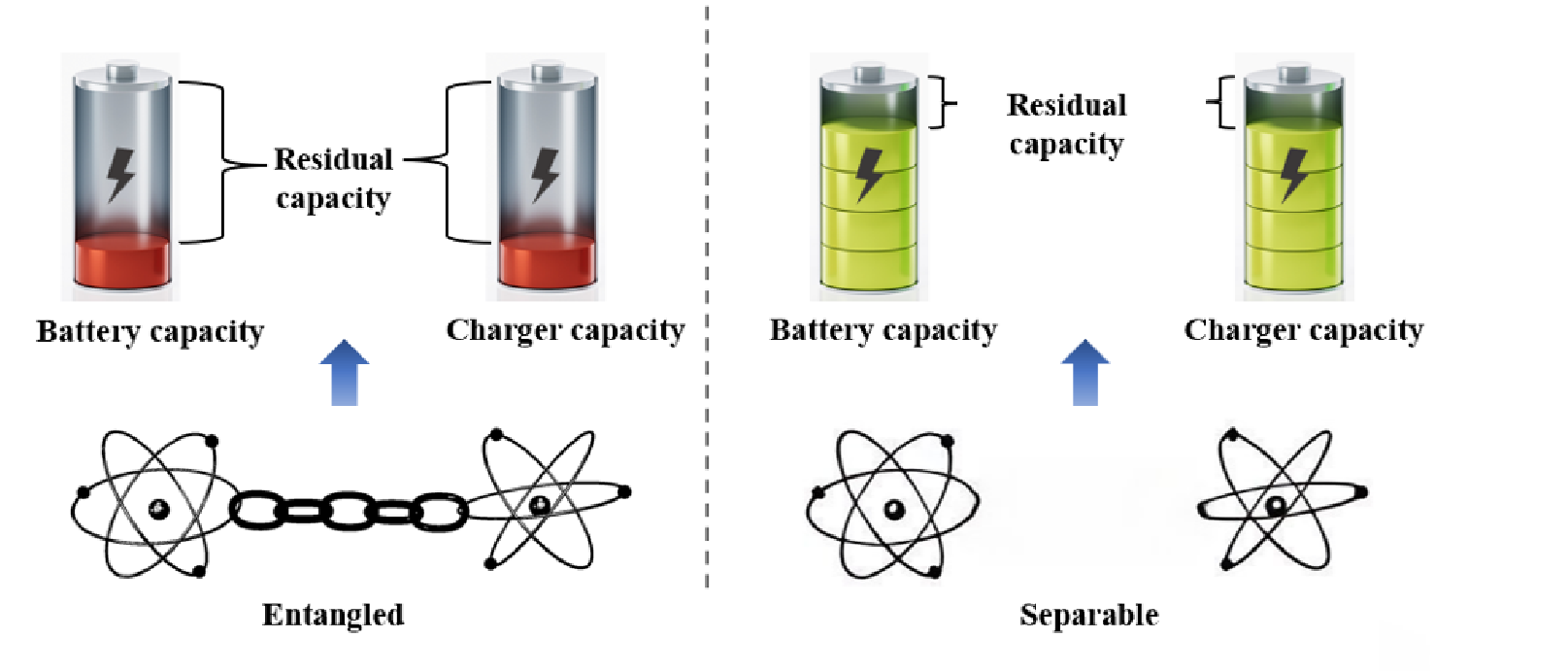}
	   	\vspace{-1em} \caption{When the battery subsystem and the charger subsystem are highly entangled, the residual battery capacity is larger, resulting in a smaller subsystem battery capacity; conversely, when the entanglement between the two subsystems is weak or they are separable, the residual battery capacity is smaller, and the subsystem battery capacity is larger.} \label{Fig.1}
	   \end{figure}
	   
	   The above results indicate that entanglement is detrimental to the capacity of the battery state. However, battery-charger entanglement plays an indispensable role in the dynamical evolution of the quantum battery system.
	   \begin{theorem}
	   During the dynamical evolution of the battery system, the battery ground state never evolves into the excited state in the absence of battery–charger entanglement, nor does the charger excited state return to its ground state. On the other hand, if the entanglement persists, the capacity of the battery state cannot reach its maximum during the evolution.
	   \end{theorem}
	   \begin{proof}
	   	The condition of vanishing entanglement implies that the interaction term $H_I$	given by Eq. (\ref{e10}) reduces to the form with $J_1=0$, i.e., $H_I=J_2\sigma_z^b\otimes\sigma_z^c$. The energy expectation value of the battery, $\langle H_b\rangle$, is governed by the Heisenberg equation,
	   	\begin{equation*}
	   	\frac{d}{dt}\langle H_b\rangle=\frac{i}{\hbar}\langle[H,H_b]\rangle.
	   	\end{equation*}
	   	It is straightforward to verify that the commutator $[H,H_b]=0$, which leads to $d\langle H_b\rangle/dt=0$. This implies that the expectation value of the battery energy is constant in time, indicating energy conservation in the battery subsystem. Similarly, the charger energy is conserved, implying that transitions cannot occur for either the battery ground state or the charger excited state. Finally, the second part of the theorem is a direct corollary of Theorem 1, and we omit the details here.
	   \end{proof}
	   During the dynamic evolution of a battery system, entanglement is essential for enabling energy flow between the battery and the charger, yet it cannot persist if the battery state capacity is to reach its peak. This apparent paradox reveals how the intertwined roles of entanglement in different aspects collectively determine the performance of quantum batteries.
	   
	   \subsection{Quantum steering and Bell nonlocality}
	   
	   Furthermore, we can use entanglement to explore the relationship between the battery state capacity and the quantum steering and Bell nonlocality involving the battery and charger spins.  For this purpose, we need to calculate the correlation matrix $T$. To simplify the calculation, we denote the time-evolved state of the total system as $|\psi(t)\rangle=\alpha|01\rangle+\beta|10\rangle$. Then one have
	   \begin{equation}\label{e20}
	   	T(\rho(t))=\left(
	   	\begin{array}{ccc}
	   		2\text{Re}(\alpha^*\beta) & -2\text{Im}(\alpha^*\beta) & 0\\
	   		2\text{Im}(\alpha^*\beta) & 2\text{Re}(\alpha^*\beta) & 0\\
	   		0 & 0 & -1\\
	   	\end{array}
	   	\right ).
	   \end{equation} 
	   Thus $T(\rho(t))T(\rho(t))^\tau$ can be written as
	   \begin{equation}\label{e21}
	   	T(\rho(t))T(\rho(t))^\tau=\left(
	   	\begin{array}{ccc}
	   		4|\alpha|^2|\beta|^2 & 0 & 0\\
	   		0 & 4|\alpha|^2|\beta|^2 & 0\\
	   		0 & 0 & 1\\
	   	\end{array}
	   	\right ),
	   \end{equation} 
	   where
	   \begin{equation*}
	   	\begin{split}
	   	\alpha&=\frac{e^{-ie_1t\xi_1-e^{-ie_2t}\xi_2}}{\xi_1-\xi_2},\,\,
	   	\beta=\frac{e^{-ie_1t-e^{-ie_2t}}}{\xi_1-\xi_2}.
	   	\end{split}
	   \end{equation*}
	   Now, we can establish a relationship between the battery state capacity $\mathcal{C}(\rho_b(t);H_b)$ and quantum steering $\mathcal{S}(\rho(t))$.
	   \begin{theorem}
	   	The relationship between the battery state capacity $\mathcal{C}(\rho_b(t);H_b)$ and quantum steering $\mathcal{S}(\rho(t))$ during the dynamic evolution of the quantum battery system is given by:
	   	\begin{equation}\label{e22}
	   	\mathcal{C}(\rho_b(t);H_b)=2\omega_b\sqrt{1-\frac{\mathcal{S}(\rho(t))}{2}}.
	   	\end{equation}
	   \end{theorem}
	   \begin{proof}
	   	The concurrence of entanglement for $\rho(t)$ can be rewritten in terms of the parameters $\alpha$ and $\beta$ as:
	   	\begin{equation*}
	   	\mathcal{E}(\rho(t))=\sqrt{2(1-|\alpha|^4-|\beta|^4)}=\sqrt{4|\alpha|^2.|\beta|^2}=2|\alpha|.|\beta|.
	   	\end{equation*}
	   	 From (\ref{e19}), the $F_3$ steering for $\rho(t)$ is 
	   	$$\mathcal{S}(\rho(t))=\text{Tr}[T(\rho(t))T(\rho(t))^\tau]-1=8|\alpha|^2|\beta|^2.
	   	$$
	   	Thus, we obtain the relation between the entanglement and the steering: $\mathcal{S}(\rho(t))=2\mathcal{E}(\rho(t))^2$. Substituting this relation into Eq. (\ref{e13}) completes the proof of the theorem.
	   \end{proof}

	   From Eq. (\ref{e21}), we can present the relationship between Bell nonlocality and the battery state capacity in the following.
	   \begin{theorem}
	   	The relation between the battery state capacity $\mathcal{C}(\rho_b(t);H_b)$ and Bell nonlocality $\mathcal{B}(\rho(t))$ during the time evolution of the quantum battery system satisfy:
	   	\begin{equation}\label{e25}
	   	\mathcal{C}(\rho_b(t);H_b)=2\omega_b\sqrt{1-\mathcal{B}(\rho(t))-\frac{\mathcal{B}(\rho(t))^2}{4}}.
	   	\end{equation}
	   \end{theorem}
	   \begin{proof}
	   	From (\ref{e20}) and (\ref{e24}), we have
	   	\begin{equation*}
	   	\begin{split}
	   	\mathcal{B}(\rho(t))&=2\sqrt{\iota_1+\iota_2}-2\\
	   	                    &=2\sqrt{1+4|\alpha|^2|\beta|^2}-2\\
	   	                    &=2\sqrt{1+\mathcal{E}(\rho(t))^2}-2.
	   	\end{split}
	   	\end{equation*}
	   	This means that $$\mathcal{E}(\rho(t))=\sqrt{(1+\frac{\mathcal{B}(\rho(t))}{2})^2-1}.$$
	   	Substituting this relation into Eq. (\ref{e13}) completes the proof.
	   \end{proof}
	   
	   The main results mentioned above indicate that whether we consider battery-charger entanglement, quantum steering, or Bell nonlocality, all of these quantum resources show a negative correlation with the battery state capacity. Using entanglement as an example, we plot the time evolution of the battery-charger entanglement $\mathcal{E}$ and battery capacity $\mathcal{C}$, as shown in Figure 2.
	   \begin{figure}[htbp]
	   	\centering
	   	\includegraphics[width=0.5\textwidth]{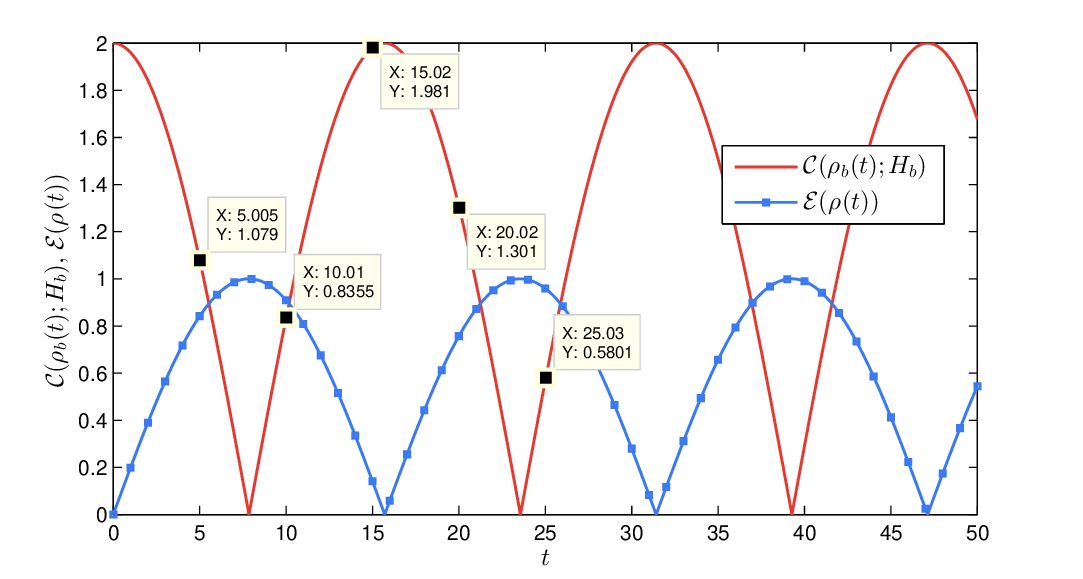}
	   	\vspace{-1em} \caption{The figure shows the relationship between the battery state capacity $\mathcal{C}(\rho_b(t);H_b)$ and the battery-charger entanglement $\mathcal{E}(\rho(t))$ during the battery system dynamic evolution. The battery state capacity decreases as the quantum entanglement. The parameters are set to $\omega_b=\omega_c=1$, $J_1=J_2=0.1$.} \label{Fig.2}  
	   \end{figure}
	   
	   To verify the analytical results, we perform independent numerical simulations. The parameters are set exactly the same as those in Figure 2: $\omega_b=\omega_c=1$, $J_1=J_2=0.1$. We directly integrate the von Neumann equation using the mesolve function in the QuTiP toolbox. The integration employs the adaptive-step Dormand–Prince algorithm over the time interval $t\in[0, 50]$ with 1000 uniformly sampled points (step size $\Delta t\approx0.05$). The table I compares the analytical and numerical values of the battery capacity at given times.
	   \begin{table}[htbp]
	   	\begin{tabular*}{0.5\textwidth}{@{}@{\extracolsep{\fill}}lllllllllllll@{}}
	   		\hline
	   		\hline
	   		Time $t$   &Analytical solution $\mathcal{C}_\text{ana}$   &Numerical solution $\mathcal{C}_\text{num}$\\
	   		\hline
	   		5.005   &~~~~1.0789   &~~~~1.0789  \\
	   		
	   		10.01   &~~~~0.8359   &~~~~0.8359  \\
	   		
	   		15.02   &~~~~1.9811   &~~~~1.9808  \\
	   		
	   		20.02   &~~~~1.3012   &~~~~1.3012  \\
	   		
	   		25.03   &~~~~0.5788   &~~~~0.5769  \\
	   		\hline
	   		\hline
	   	\end{tabular*}
	   	\caption{The times $t$ are taken from the numerical sampling points, analytical values are computed from Eq. (11), and numerical values are obtained via QuTiP simulations.}
	   \end{table}
	   
	   \subsection{Quantum coherence and imaginarity}
	   Recall that when discussing the correlation matrix, we parameterized the total system’s evolved state as $|\psi(t)\rangle=\alpha|01\rangle+\beta|10\rangle$, where
	   \begin{equation*}
	   	\begin{split}
	   		\alpha&=\frac{e^{-ie_1t\xi_1-e^{-ie_2t}\xi_2}}{\xi_1-\xi_2},\,\,
	   		\beta=\frac{e^{-ie_1t-e^{-ie_2t}}}{\xi_1-\xi_2}.
	   	\end{split}
	   \end{equation*}
	   Then, the coherent $l_1$ norm of the time-evolved state $\rho(t)=|\psi(t)\rangle\langle\psi(t)|$ can be written as:
	   \begin{equation*}
	   C_1(\rho(t))=\sum_{i\neq j}|\rho(t)_{ij}|=|\alpha\beta^*|+|\alpha^*\beta|=2|\alpha|.|\beta|=\mathcal{E}(\rho(t)).
	   \end{equation*}
	   This implies that the coherent $l_1$-norm is numerically equal to the entanglement concurrence $\mathcal{E}$ during the dynamical evolution of the quantum battery system. Thus, according to Theorem 1, we directly obtain the trade-off between quantum coherence and the battery state capacity.
	   \begin{theorem}
	   	The trade-off between the battery state capacity $\mathcal{C}(\rho_b(t);H_b)$ and quantum coherence $C_1(\rho(t))$ during the time evolution of the quantum battery system is given by
	   	\begin{equation}\label{c2}
	   	\mathcal{C}(\rho_b(t);H_b)=2\omega_b\sqrt{1-C_1(\rho(t))^2}.
	   	\end{equation}
	   \end{theorem}
	   
	   	Under a given basis, quantum coherence is characterized by the presence of off-diagonal elements, while imaginarity is a subset of coherence that further requires these off-diagonal elements to have complex phases. All states with imaginarity possess coherence, but the converse is not necessarily true, making imaginarity a more specific quantum resource than coherence. Complex numbers are indispensable in quantum mechanics, and the discovery of experimental scenarios \cite{mord} that cannot be accurately modeled using only real-valued quantum theories has prompted increased attention to the study of imaginarity. Here, we establish a trade-off relation between battery state capacity $\mathcal{C}(\rho_b(t);H_b)$ and quantum imaginarity $\mathcal{I}(\rho(t))$ during the dynamical evolution of the battery system.
	   \begin{theorem}
	   	During the temporal evolution of the battery system, the battery state capacity and quantum imaginarity satisfy the following trade-off relation:
	   	\begin{equation}\label{e30}
	   		\mathcal{C}(\rho_b(t);H_b)=2\omega_b\sqrt{1-4\text{Re}(\alpha\beta^*)^2-\mathcal{I}(\rho(t))^2},
	   	\end{equation}
	   	where
	   	\begin{equation*}
	   		\text{Re}(\alpha\beta^*)=\frac{\frac{(\omega_b-\omega_c)}{J_1}(1-\cos(e_1-e_2)t)}{\frac{2[J_1^2+(\omega_b-\omega_c)^2]}{J_1^2}}.
	   	\end{equation*}
	   \end{theorem}
	   \begin{proof}
	   	From the expression for the evolved state of the total system, $\rho(t)=|\psi(t)\rangle\langle\psi(t)|$, the quantum imaginarity $\mathcal{I}(\rho(t))$ can be expressed as
	   	\begin{equation*}
	   		\begin{split}
	   			\mathcal{I}(\rho(t))&=\sum_{i\neq j}|\text{Im}(\rho(t)_{ij})|\\
	   			&=2\frac{\text{Im}[(e^{-ie_1t}\xi_1-e^{-ie_2t}\xi_2)(e^{ie_1t}-e^{ie_2t})]}{(\xi_1-\xi_2)^2}\\
	   			&=\frac{\frac{\sqrt{J_1^2+(\omega_b-\omega_c)^2}}{J_1}|\sin(e_1-e_2)t|}{1+\frac{(\omega_b-\omega_c)^2}{J_1^2}}.
	   		\end{split}
	   	\end{equation*}
	   	Next, we explore its link to quantum coherence,
	   	\begin{equation*}
	   		\begin{split}
	   			|\alpha|^2|\beta|^2&=(\alpha\beta^*).(\alpha^*\beta)=(\alpha\beta^*).(\alpha\beta^*)^*\\
	   			&=(\text{Re}(\alpha\beta^*)+i\text{Im}(\alpha\beta^*)).(\text{Re}(\alpha\beta^*)-i\text{Im}(\alpha\beta^*))\\
	   			&=\text{Re}(\alpha\beta^*)^2+\text{Im}(\alpha\beta^*)^2\\
	   			&=\text{Re}(\alpha\beta^*)^2+|\text{Im}(\alpha\beta^*)|^2\\
	   			&=\text{Re}(\alpha\beta^*)^2+\frac{\mathcal{I}(\rho(t))^2}{4}.
	   		\end{split}
	   	\end{equation*}
	   	This means that 
	   	\begin{equation}\label{e31}
	   		C_1(\rho(t))=2\sqrt{\text{Re}(\alpha\beta^*)^2+\frac{\mathcal{I}(\rho(t))^2}{4}}.
	   	\end{equation}
	   	By calculation we can obtain that
	   	\begin{equation*}
	   		\text{Re}(\alpha\beta^*)=\frac{\frac{(\omega_b-\omega_c)}{J_1}(1-\cos(e_1-e_2)t)}{\frac{2[J_1^2+(\omega_b-\omega_c)^2]}{J_1^2}}.
	   	\end{equation*}
	   	Finally, substituting Eq. (\ref{e31}) into Eq. (\ref{c2}) completes the proof of the theorem.
	   \end{proof}
	   As evident from Eq. (\ref{e30}), an increase in the imaginarity resource necessarily leads to a decrease in the battery state capacity. This can be understood by noting that imaginarity here serves to enhance quantum coherence. In other words, under these circumstances, imaginarity constitutes a component of coherence. 
	   
	   In addition, while imaginarity aligns with entanglement, steering, Bell nonlocality, and coherence in adversely affecting the battery state capacity, it differs from the latter four in that capacity does not necessarily peak when imaginarity vanishes. This actually depends on whether the battery system is detuned or resonant. Setting $\mathcal{I}(\rho(t))=0$ in Eq. (\ref{e30}) yields
	   \begin{equation*}
	   	\mathcal{C}(\rho_b(t);H_b)=2\omega_b\sqrt{1-[\frac{J_1(\omega_b-\omega_c)(1-\cos(e_1-e_2)t)}{[J_1^2+(\omega_b-\omega_c)^2]}]^2}
	   \end{equation*}
	   When the system is resonant ($\omega_b=\omega_c$), the off-diagonal elements of the density matrix $\rho(t)$ during dynamical evolution are purely imaginary. Vanishing imaginarity implies vanishing coherence, allowing the battery capacity to reach its peak value of $2\omega_b$. In contrast, when the system is detuned ($\omega_b\neq\omega_c$), the off-diagonal elements are no longer purely imaginary, and the coherence contains the contributions from real part of the off-diagonal elements. Consequently, the battery capacity cannot attain the maximum value of $2\omega_b$. Furthermore, we examine the dynamics of the battery capacity $\mathcal{C}(\rho_b(t);H_b)$, the quantum coherence $C_1(\rho(t))$, and the imaginarity $\mathcal{I}(\rho(t))$ for system detunings $\Delta=|\omega_b-\omega_c|=0.2, 0.3, 0.4$, and 0.5, as shown in Figure 3.

	   \begin{figure*}[htbp]
	   	\centering
	   	\includegraphics[width=1\textwidth]{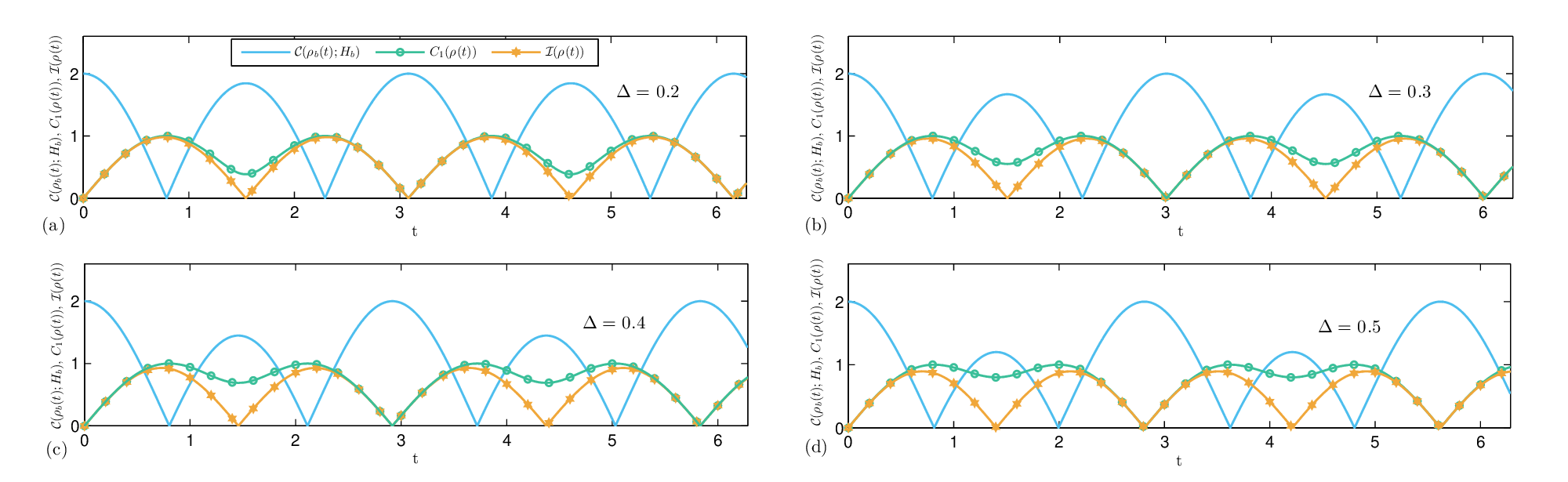}
	   	\vspace{-1em} \caption{Panels (a), (b), (c), and (d) show the dynamics of the battery capacity $\mathcal{C}(\rho_b(t);H_b)$, the quantum coherence $C_1(\rho(t))$, and the imaginarity $\mathcal{I}(\rho(t))$ for system detunings $\Delta=0.2, 0.3, 0.4$, and 0.5, respectively. The other parameters are set to $\omega_b=J_1=J_2=1$.} \label{Fig.3}
	   \end{figure*}

	   \subsection{Quantum state texture}
	   In contrast to entanglement, the time evolution of the QST matches precisely that of the battery capacity. We can also investigate the mathematical link between the battery capacity and the QST. 
		 \begin{theorem}
		 During the time evolution of the quantum battery system, the QST measure $\mathcal{T}_{tr}(\rho_b(t))$ and the battery state capacity $\mathcal{C}(\rho_b(t);H_b)$ at time $t$ satisfy:
		 \begin{equation}\label{e26}
		 	\mathcal{C}(\rho_b(t);H_b)=2\omega_b\sqrt{4\mathcal{T}_{tr}(\rho_b(t))^2-1}.
		 \end{equation}
		 \end{theorem}
		 \begin{proof}
		 For simplicity, we adopt the same parametrization for $\rho_b(t)$ as used in the proof of Theorem 1. Thus, the battery capacity is given by $\mathcal{C}(\rho_b(t);H_b)=2\omega_b|1-2p(t)|$. Eq. (\ref{e3}) indicates that $\mathcal{T}_{tr}(\rho_b(t))=\sqrt{(1-2p(t))^2+1}/2$. This yields a relation between the battery capacity and the QST:
		 \begin{equation*}
		 \begin{split}
		 \mathcal{T}_{tr}(\rho_b(t))&=\frac{\sqrt{(1-2p(t))^2+1}}{2}=\frac{\sqrt{|1-2p(t)|^2+1}}{2}\\
		 &=\frac{1}{2}\sqrt{\frac{\mathcal{C}(\rho_b(t);H_b)^2}{4\omega_b^2}+1}.
		 \end{split}
		 \end{equation*}
		 \end{proof}
		 The QST measure $\mathcal{T}_{tr}$ quantifies the roughness or uneveness of the density matrix, which is directly related to the eigenvalue distribution of the state. Positive semidefiniteness is a fundamental property of density matrices. Based on this, we say that $\rho$ is majorized by $\sigma$ ($\rho\prec\sigma$) if
		 \begin{equation*}
		 	\begin{split}
		 		&\sum_{i=1}^{k}\lambda_{d-i+1}<\sum_{i=1}^{k}\eta_{d-i+1},\,1\leq k<d.\\
		 		&\sum_{i=1}^{d}\lambda_{d-i+1}=\sum_{i=1}^{d}\eta_{d-i+1}.
		 	\end{split}
		 \end{equation*} 
		 Here, $\{\lambda_i\}$ and $\{\eta_i\}$ denote the eigenvalues of $\rho$ and $\sigma$, respectively, arranged in ascending order, and $d$ is their shared dimension. The more uneven the eigenvalue distribution of a quantum state, the stronger the majorization. It can be verified that the majorization for different eigenvalue distributions satisfies the following relation:
		 \begin{equation*}
		 (1,0,\ldots,0)\succ(\lambda_d,\lambda_{d-1},\ldots,\lambda_1)\succ(\frac{1}{d},\frac{1}{d},\ldots,\frac{1}{d}),
		 \end{equation*}
		 where $\{\lambda_i\}$ represents any eigenvalue distribution distinct from the two distributions on the left and right. In the dynamical evolution of the battery state, a high degree of QST implies an uneven eigenvalue distribution, which in turn leads to enhanced majorization of the battery state $\rho_b(t)$. Since the battery capacity $\mathcal{C}$ is a Schur-convex functional, greater majorization results in a higher battery capacity. The relationship between the QST of the battery state and its capacity during the dynamical evolution, is shown in Figure 4.
		 \begin{figure}[htbp]
		 	\centering
		 	\includegraphics[width=0.5\textwidth]{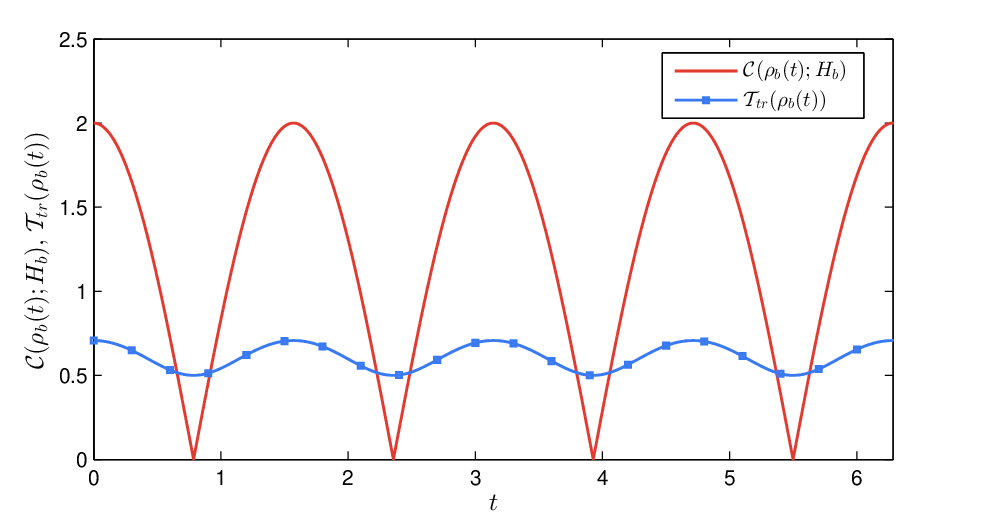}
		 	\vspace{-1em} \caption{During the time evolution of the battery state, the battery capacity $\mathcal{C}(\rho_b(t);H_b)$ increases with the degree of QST.} \label{Fig.4}
		 \end{figure}
		 
		 Additionally, we can investigate the relationship between the QST degree of the battery state and its residual battery capacity.
		 \begin{theorem}
		 During the time evolution of the quantum battery system governed by the Hamiltonian in Eqs. (\ref{e7}-\ref{e10}), the QST measure $\mathcal{T}_{tr}(\rho_b(t))$ and the residual battery capacity $R(\rho(t),H)$ at time $t$ satisfy:
		 \begin{equation}\label{e27}
		 4\mathcal{T}_{tr}(\rho_b(t))^2=1+\frac{(\epsilon_4-\epsilon_1-R(\rho(t),H))^2}{4(\omega_b+\omega_c)^2},
		 \end{equation}
		 where $\epsilon_4$ and $\epsilon_1$ represent the maximum and minimum eigenvalues of the total system Hamiltonian, respectively.
		 \end{theorem}
		 \begin{proof}
		 According to the proof of Theorem 3, the residual battery capacity at time $t$ is given by
		 \begin{equation}\label{e28}
		 	R(\rho(t),H)=\epsilon_4-\epsilon_1-2(\omega_b+\omega_c)|1-2p(t)|.
		 \end{equation}
		 The expression for the QST measure $\mathcal{T}_{tr}$ of the battery state can be written as
		 \begin{equation}\label{e29}
		 	\mathcal{T}_{tr}(\rho_b(t))=\frac{1}{2}\sqrt{1+(1-2p(t))^2}.
		 \end{equation}
		 Combining Eqs. (\ref{e28}) and (\ref{e29}), we can derive the relation between the residual battery capacity and the QST of the battery state during the evolution.
		 \end{proof}
		 Eq. (\ref{e27}) clearly shows that $R(\rho(t),H)$ is negatively correlated with $\mathcal{T}_{tr}(\rho_b(t))$. Therefore, the increase in battery state capacity with the degree of QST can also be understood from the perspective of the trade-off relationship between subsystem battery capacity and residual battery capacity. Furthermore, based on the above results and analysis, we can also establish the following relationship between QST and entanglement, steering, Bell nonlocality, coherence, imaginarity.
		 \begin{equation*}
		 \begin{split}
		 &\mathcal{E}(\rho(t))=2\sqrt{\frac{1}{2}-\mathcal{T}_{tr}(\rho_b(t))^2},\\
		 &\mathcal{S}(\rho(t))=4-8\mathcal{T}_{tr}(\rho_b(t))^2,\\
		 &\mathcal{B}(\rho(t))=2\sqrt{3-4\mathcal{T}_{tr}(\rho_b(t))^2}-2,\\
		 &C_1(\rho(t))=2\sqrt{\frac{1}{2}-\mathcal{T}_{tr}(\rho_b(t))^2},\\
		 &\mathcal{I}(\rho(t))=\sqrt{2-4\text{Re}(\alpha\beta^*)^2-4\mathcal{T}_{tr}(\rho_b(t))},
		 \end{split}
		 \end{equation*}
		  where
		  $$\text{Re}(\alpha\beta^*)=\frac{\frac{(\omega_b-\omega_c)}{J_1}(1-\cos(e_1-e_2)t)}{\frac{2[J_1^2+(\omega_b-\omega_c)^2]}{J_1^2}}.$$ Since we introduce no additional assumptions regarding the system parameters or evolution time, the relationships we obtain, including those between battery capacity and quantum resources and those among the quantum resources themselves, are independent of the choice of system parameters and evolution time. Additionally, we discuss the relationship between quantum resources and battery capacity under dephasing noise channel in Appendix B. Although this relationship is also affected by the noise strength, it remains monotonic for a fixed noise strength. Since dephasing noise does not alter the populations of the battery states, the connection between quantum resource and battery capacity remains unaffected.

	\section{IV. Conclusions and discussions}
    In this work, we have investigated the correlations between quantum battery capacity and quantum resources. It have been observed that the battery capacity decreased monotonically as quantum entanglement, steering, Bell nonlocality, coherence increased, with the capacity reaching its maximum only in the absence of these four quantum resources. To further elucidate this behavior, we have introduced the residual battery capacity, defined as the gap between the total system capacity and the sum of the individual spin capacities of the battery and charger. This residual capacity have been found to exhibit a positive correlation with the degree of entanglement. Moreover, unlike first four quantum resources, although quantum imaginarity has also led to a monotonic reduction in battery capacity, its vanish under system detuning did not necessarily coincide with peak capacity. This effect had become increasingly pronounced with larger detuning values. In contrast to entanglement, steering, Bell nonlocality, coherence, and imaginarity, quantum state texture has demonstrated a positive correlation with battery capacity, while showing a negative relationship with the aforementioned quantum resources and the residual capacity. 
    
Conventional wisdom holds that quantum resources are fundamental to the superiority of quantum batteries. However, our work has revealed the dual role of quantum entanglement: while it is essential for energy transfer between the charger and battery, it simultaneously limits the battery state capacity. Notably, we have identified quantum state texture as a distinct resource which, unlike other resources, exerts a positive influence on capacity. Further investigation of quantum state texture has may emerged as a key direction for enhancing future quantum energy storage systems.
    
    \bigskip
{\bf Conflict of Interest}

The authors declare no conflict of interest.

{\bf Data Availability Statement}

Data sharing is not applicable to this article as no new data were created or analyzed in this study.

{\bf Keywords} 
  quantum battery capacity; quantum entanglement; Bell nonlocality; quantum coherence; quantum state texture

	\bigskip
	{\bf Acknowledgments:} ~This work is supported by the National Natural Science Foundation of China (NSFC) under Grant Nos. 12204137 and  12564048; the Natural Science Foundation of Hainan Province under Grant No. 125RC744 and the specific research fund of the Innovation Platform for Academicians of Hainan Province.

	\appendix
	\section{Appendix A: Proof of Eq. (\ref{e14})}
	\setcounter{equation}{0}
	\renewcommand{\theequation}{A\arabic{equation}}
    Since the density matrix of the entire battery system remains of the X-form throughout the dynamical evolution, it is sufficient to restrict $\rho$ to the X-states to complete the proof. 
    
    Given the two-qubit X state $\rho$ in the computational basis,
    $$
    \rho=\left(\begin{array}{cccc}
    	\rho_{11} & 0 & 0 & \rho_{14} \\
    	0 & \rho_{22} & \rho_{23} & 0 \\
    	0 & \rho_{32} & \rho_{33} & 0 \\
    	\rho_{41} & 0 & 0 & \rho_{44}
    \end{array}\right),
    $$
    where $\sum_{i=1}^{4}\rho_{ii}=1$ and $\rho_{ii}\geq0$. We denote the dephased state of $\rho$ by $\Delta(\rho)$. Note that $\Delta(\rho)$ is a diagonal state, so its eigenvalues are its populations. We arrange these eigenvalues in ascending order and denote them as $\{\alpha_i\}_{i=1}^4$. For the whole system Hamiltonian $H=H_b+H_c+H_I$ with eigenvalues $$\{-J_2\pm\sqrt{J_1^2+(\omega_b-\omega_c)^2}, J_2\pm(\omega_b+\omega_c)\}\equiv\{\epsilon_i\}_{i=1}^4$$ in ascending order, one have (without loss of generality, assume $\omega_b\geq\omega_c$)
    \begin{widetext}
    \begin{equation*}
    \begin{split}
    \mathcal{C}(\Delta(\rho);H)&=(\alpha_4-\alpha_1)(\epsilon_4-\epsilon_1)+(\alpha_3-\alpha_2)(\epsilon_3-\epsilon_2)\geq(\alpha_4-\alpha_1)[J_2+(\omega_b+\omega_c)-(J_2-(\omega_b+\omega_c))]\\
    +&(\alpha_3-\alpha_2)[-J_2+\sqrt{J_1^2+(\omega_b-\omega_c)^2}-(-J_2-\sqrt{J_1^2+(\omega_b-\omega_c)^2})]\geq2(\alpha_4-\alpha_1)(\omega_b+\omega_c)+2(\alpha_3-\alpha_2)(\omega_b-\omega_c)\\
    =&2\omega_b(\alpha_4+\alpha_3-\alpha_2-\alpha_1)+2\omega_c(\alpha_4+\alpha_2-\alpha_3-\alpha_1)\geq2\omega_b|\rho_{11}+\rho_{22}-\rho_{33}-\rho_{44}|+2\omega_c|\rho_{11}+\rho_{33}-\rho_{22}-\rho_{44}|\\
    =&\mathcal{C}(\rho_b;H_b)+\mathcal{C}(\rho_c;H_c).
    \end{split}
    \end{equation*}
    \end{widetext}
    The first inequality holds because the $\{\epsilon_i\}_{i=1}^4$ are arranged in ascending order of the Hamiltonian eigenvalues. The second inequality stems from a straightforward bounding relation: $\sqrt{J_1^2+(\omega_b-\omega_c)^2}\geq(\omega_b-\omega_c)$. The third inequality is due to the fact that the $\{\alpha_i\}_{i=1}^4$ are an ascending rearrangement of the populations $\{\rho_{ii}\}_{i=1}^4$. Observing the majorization relation $\Delta(\rho)\prec\rho$, together with the Schur-convexity of the battery capacity \cite{yyas}, we obtain that:
    \begin{equation*}
    \mathcal{C}(\rho;H)\geq\mathcal{C}(\Delta(\rho);H)\geq\mathcal{C}(\rho_b;H_b)+\mathcal{C}(\rho_c;H_c).
    \end{equation*}
    
   \section{Appendix B: Under the dephasing noise channel}
   \setcounter{equation}{0}
   \renewcommand{\theequation}{B\arabic{equation}}
   The dephasing noise channel can be modeled by the following set of Kraus operators:
   \begin{equation}
   	\begin{split}
   	&K_1=(1-\gamma)I\otimes I,\,K_2=\sqrt{\gamma(1-\gamma)}I\otimes\sigma_z,\\
   	&K_3=\sqrt{\gamma(1-\gamma)}\sigma_z\otimes I,\,K_4=\gamma\sigma_z\otimes\sigma_z,
   	\end{split}
   \end{equation}
   where $\gamma\in[0,1]$ is the phase flip probability. Denoting the time-evolved state of the bipartite system as $|\psi(t)\rangle=\alpha|01\rangle+\beta|10\rangle$, one has
   \begin{equation*}
   \begin{split}
   \tilde{\rho}(t)&=\Upsilon(|\psi(t)\rangle\langle\psi(t)|)=\sum_{i=1}^{4}K_i|\psi(t)\rangle\langle\psi(t)|K_i^\dagger.
   \end{split}
   \end{equation*}
   Then the density matrix can be written as 
   $$
   \tilde{\rho}(t)=\left(\begin{array}{cccc}
   	0 & 0 & 0 & 0 \\
   	0 & |\alpha|^2 & (1-2\gamma)^2\alpha\beta^* & 0 \\
   	0 & (1-2\gamma)^2\alpha^*\beta & |\beta|^2 & 0 \\
   	0 & 0 & 0 & 0
   \end{array}\right).
   $$
   Dephasing noise only affects the coherence terms while preserving the population of the quantum states. Therefore, the capacity of the battery state remains unchanged. The battery–charger entanglement can, at this point, be written as \cite{smhr}:
   \begin{equation}
   \mathcal{E}(\tilde{\rho}(t))=2(1-2\gamma)^2|\alpha||\beta|.
   \end{equation}
   For $\gamma=\frac{1}{2}$, there is no entanglement between the battery and charger. When $\gamma\neq\frac{1}{2}$, the capacity of the battery state is related to the battery-charger entanglement by the expression:
   \begin{equation}
   	\mathcal{C}(\tilde{\rho}_b(t);H_b)=2\omega_b\sqrt{1-\frac{\mathcal{E}(\tilde{\rho}(t))^2}{(1-2\gamma)^4}}.
   \end{equation}
   Despite the inclusion of the additional noise parameter $\gamma$ in the relation, the battery capacity and entanglement still exhibit a monotonically decreasing relationship for a fixed $\gamma$. Additionally, in the absence of noise, the expression reduces to Eq. (\ref{e13}).
   
   Now, We turn to quantum steering and Bell nonlocality. To this end, it is necessary to compute the product of the correlation matrix with its transpose under a dephasing noise channel, which can be written as
   \begin{small}
   	\begin{equation*}
   		T(\tilde{\rho}(t))T(\tilde{\rho}(t))^\tau=\left(
   		\begin{array}{ccc}
   			4(1-2\gamma)^4|\alpha|^2|\beta|^2 & 0 & 0\\
   			0 & 4(1-2\gamma)^4|\alpha|^2|\beta|^2 & 0\\
   			0 & 0 & 1\\
   		\end{array}
   		\right ),
   	\end{equation*} 
   \end{small}
   where
   \begin{equation*}
   	\begin{split}
   		\alpha&=\frac{e^{-ie_1t\xi_1-e^{-ie_2t}\xi_2}}{\xi_1-\xi_2},\,\,
   		\beta=\frac{e^{-ie_1t-e^{-ie_2t}}}{\xi_1-\xi_2}.
   	\end{split}
   \end{equation*}
   Based on this, we have
   \begin{equation}
   \begin{split}
   &\mathcal{S}(\tilde{\rho}(t))=8(1-2\gamma)^4|\alpha|^2|\beta|^2,\\
   &\mathcal{B}(\tilde{\rho}(t))=2\sqrt{1+4(1-2\gamma)^4|\alpha|^2|\beta|^2}-2.
   \end{split}
   \end{equation}
   When $\gamma\neq\frac{1}{2}$, the relations between the capacity of the battery state and quantum steering, as well as Bell nonlocality, are
   \begin{equation*}
   \begin{split}
   \mathcal{C}(\tilde{\rho}_b(t);H_b)&=2\omega_b\sqrt{1-\frac{\mathcal{S}(\tilde{\rho}(t))}{2(1-2\gamma)^4}},\\
   \mathcal{C}(\tilde{\rho}_b(t);H_b)&=2\omega_b\sqrt{1-\frac{1}{(1-2\gamma)^4}(\mathcal{B}(\tilde{\rho}(t))+\frac{\mathcal{B}(\tilde{\rho}(t))^2}{4})}.
   \end{split}
   \end{equation*}
   
   Next, we discuss quantum coherence and imaginarity. Note that for the state $\tilde{\rho}(t)$, the $l_1$-norm of coherence is equal to the entanglement concurrence, i.e.,
   \begin{equation*}
   	C_1(\tilde{\rho}(t))=\sum_{i\neq j}|\tilde{\rho}(t)_{ij}|=2(1-2\gamma)^2|\alpha|.|\beta|=\mathcal{E}(\tilde{\rho}(t)).
   \end{equation*}
   Using the relation between coherence and imaginarity,
   \begin{equation}\label{e31}
   	C_1(\tilde{\rho}(t))^2=4(1-2\gamma)^4\text{Re}(\alpha\beta^*)^2+\mathcal{I}(\tilde{\rho}(t))^2.
   \end{equation}
   one can obtain the following results:
   \begin{equation*}
   	\begin{split}
   	\mathcal{C}(\tilde{\rho}_b(t);H_b)&=2\omega_b\sqrt{1-C_1(\tilde{\rho}(t))^2},\\
   	\mathcal{C}(\tilde{\rho}_b(t);H_b)&=2\omega_b\sqrt{1-4(1-2\gamma)^4\text{Re}(\alpha\beta^*)^2-\mathcal{I}(\tilde{\rho}(t))^2},
   	\end{split}
   \end{equation*}
   where
   \begin{equation*}
   \text{Re}(\alpha\beta^*)=\frac{\frac{(\omega_b-\omega_c)}{J_1}(1-\cos(e_1-e_2)t)}{\frac{2[J_1^2+(\omega_b-\omega_c)^2]}{J_1^2}}.
   \end{equation*}
   
   Under the influence of a dephasing noise channel, although the quantum state texture of the composite system changes, the quantum state texture of the battery subsystem remains unchanged. Consequently, the relationship between the capacity of the battery state and its quantum state texture is unaffected by the noise parameter $\gamma$, exhibiting complete robustness in this scenario.
\end{document}